\documentclass[11pt]{article}
\usepackage[margin=1in]{geometry}
\usepackage{latexsym, amscd, amsfonts, mathrsfs, amsmath, amssymb, amsthm}
\usepackage{bm}
\usepackage{hyperref,color}
\usepackage[noend]{algpseudocode}
\usepackage[capitalize,nameinlink]{cleveref}

\makeatletter
\newcounter{algorithmicH}
\let\oldalgorithmic\algorithmic
\renewcommand{\algorithmic}{%
  \stepcounter{algorithmicH}
  \oldalgorithmic}
\renewcommand{\theHALG@line}{ALG@line.\thealgorithmicH.\arabic{ALG@line}}
\makeatother

\usepackage[dvipsnames]{xcolor}
\hypersetup{
	colorlinks=true,
	pdfpagemode=UseNone,
    citecolor=OliveGreen,
    linkcolor=NavyBlue,
    urlcolor=Magenta,
	pdfstartview=FitW
}
\usepackage{stmaryrd, tikz-cd, enumitem, mathrsfs, bbm, algorithm, url, esint, todonotes, listings, moreverb, makecell}

\def\bE{\mathbb{E}}

\def\bR{\mathbb{R}}
\def\bN{\mathbb{N}}

\def\cD{\mathcal{D}}

\def\cI{\mathcal{I}}

\def\cM{\mathcal{M}}

\renewcommand{\b}{{\bm{b}}}
\newcommand{\f}{\bm{f}}
\newcommand{\blam}{\bm{\lambda}}

\newcommand{\st}{S}

\def\wt{\widetilde}

\renewcommand{\setminus}{\backslash}

\newcommand{\rom}[1]{\textup{\uppercase\expandafter{\romannumeral#1}}}

\newtheorem{theorem}{Theorem}
\newtheorem{lemma}[theorem]{Lemma}
\newtheorem{proposition}[theorem]{Proposition}

\newtheorem{claim}[theorem]{Claim}

\newtheorem{observation}[theorem]{Observation}

\theoremstyle{definition}
\newtheorem{definition}[theorem]{Definition}

\newtheorem{remark}[theorem]{Remark}

\begin{document}
\title{Fast Sampling of $b$-Matchings and $b$-Edge Covers}
 \author{Zongchen Chen\thanks{\texttt{zongchen@mit.edu}. MIT.}
 \and
 Yuzhou Gu\thanks{\texttt{yuzhougu@mit.edu}. MIT.}}
\date{}

\maketitle

\begin{abstract}
For an integer $b \ge 1$, a $b$-matching (resp.~$b$-edge cover) of a graph $G=(V,E)$ is a subset $S\subseteq E$ of edges such that every vertex is incident with at most (resp.~at least) $b$ edges from $S$.
We prove that for \emph{any} $b \ge 1$ the simple Glauber dynamics for sampling (weighted) $b$-matchings and $b$-edge covers mixes in $O(n\log n)$ time on all $n$-vertex bounded-degree graphs. This significantly improves upon previous results which have worse running time and only work for $b$-matchings with $b \le 7$ and for $b$-edge covers with $b \le 2$.


More generally, we prove \emph{spectral independence} for a broad class of binary symmetric Holant problems with \emph{log-concave} signatures, including $b$-matchings, $b$-edge covers, and antiferromagnetic $2$-spin edge models.
We hence deduce optimal mixing time of the Glauber dynamics from spectral independence.

The core of our proof is a recursive coupling inspired by \cite{chen2023near} which upper bounds the Wasserstein $W_1$ distance between distributions under different pinnings. Using a similar method, we also obtain the optimal $O(n\log n)$ mixing time of the Glauber dynamics for the hardcore model on $n$-vertex bounded-degree claw-free graphs, for \emph{any} fugacity $\lambda$. This improves over previous works which have at least cubic dependence on $n$.

\end{abstract}

\section{Introduction} \label{sec:intro}
\subsection{\texorpdfstring{$b$}{b}-Matchings and \texorpdfstring{$b$}{b}-edge covers}

Let $G=(V,E)$ be a graph and $b\ge 1$ be an integer.
Let $E_v = \{e \in E: e \text{~incident to~} v\}$ be the set of all adjacent edges of a vertex $v \in V$.
A $b$-matching of $G$ is a subset $S \subseteq E$ of edges such that $|S\cap E_v| \le b$ for all $v \in V$.
When $b = 1$ this reduces to a usual matching of $G$.
We consider the problem of sampling random weighted $b$-matchings of a given graph $G$.
Write $\cM_b = \cM_b(G)$ for the collection of all $b$-matchings of $G$.
For $\lambda>0$, consider the Gibbs distribution $\mu = \mu_{G,b,\lambda}$ on $\cM_b$ given by
\[
\mu(S) := \frac{\lambda^{|S|}}{Z}, \quad \forall S \in \cM_b
\]
where $Z = Z_{G,b}(\lambda)$ is a normalization constant, known as the partition function, defined as
\[
Z := \sum_{S \in \cM_b} \lambda^{|S|}.
\]
Note that if $\lambda=1$ then $\mu$ is the uniform distribution over $\cM_b$ and $Z$ counts the total number of $b$-matchings in $\cM_b$.

For $b = 1$, namely the usual matchings, such a model is called the \emph{monomer-dimer model}.
Approximately counting and sampling matchings is a fundamental problem in theoretical computer science and also one of the first successful applications of Markov chain Monte Carlo (MCMC) methods in approximate sampling and counting combinatorial objects.
In a classical work \cite{JS89}, Jerrum and Sinclair proved rapid mixing of Glauber dynamics for sampling from the monomer-dimer model.
The besting mixing time result to date is $O(n^2m\log n)$ on arbitrary graphs where $m$ is the number of edges \cite{Jbook}, and only very recently this was improved to $O(n\log n)$ on all bounded-degree graphs \cite{chen2022spectral}.

For general $b \ge 1$, \cite{HLZ16} presented a polynomial-time algorithm for approximately sampling $b$-matchings on all graphs when $b \le 7$. Their algorithm is based on MCMC and they utilize the notion of \emph{windable functions} introduced in \cite{McQ13} to construct canonical paths and bound the spectral gap of the Markov chain.
However, as pointed out in \cite{HLZ16}, for $8$-matchings the associated constraint function is no longer windable under their characterization and hence their approach cannot work for $b \ge 8$.

Another closely related problem is sampling $b$-edge covers of a given graph.
A subset $S \subseteq E$ of edges is called a $b$-edge cover if every vertex is incident with at least $b$ edges, i.e., $|S\cap E_v| \ge b$ for every $v \in V$.
For $b=1$, i.e., the usual edge covers, the counting and sampling problems have been extensively studied as well \cite{BR09,LLL14,LLZ14,HLZ16,GLLZ21,BCR21,chen2022spectral}.
In particular, \cite{LLL14} presented a deterministic algorithm for counting unweighted edge covers for all graphs using the correlation decay approach with a running time $O(m^{1+\log_2 6}n^2)$, and this was later generalized to weighted edge covers in \cite{LLZ14}.
Deterministic algorithms based on the polynomial interpolation approach were also given for all bounded-degree graphs in \cite{GLLZ21,BCR21}.
More recently, it was shown in \cite{chen2022spectral} that the Glauber dynamics for sampling edge covers mixes in $O(n\log n)$ time on all bounded-degree graphs.

Meanwhile, the problem of sampling and counting $b$-edge covers for larger $b$ is far from clear.
The MCMC-based algorithm in \cite{HLZ16} can be applied to count $b$-edge covers for $b \le 2$, which only slightly extends the classical case of $b=1$. Similar to $b$-matchings, the approach from \cite{HLZ16} no longer works for $b\ge 3$ due to the failure of windability.

In this paper we attempt to answer the following question: Are there polynomial-time algorithms for approximately sampling/counting $b$-matchings and $b$-edge covers of a given graph for \emph{any} $b \ge 1$?
We give a positive answer to this question for all bounded-degree graphs.
More specifically, we show that the Glauber dynamics, a simple Markov chain for sampling $b$-matchings/$b$-edge covers, converges in $O(n\log n)$ time which is optimal.

One can simultaneously generalize both $b$-matchings and $b$-edge covers by assigning a different threshold to each vertex.
More specifically, let $\b = (b_v)_{v \in V} \in \bN^V$ be a vector of thresholds on all vertices.
We consider the collection $\cM_\b = \cM_\b(G)$ of generalized $\b$-matchings, defined as
\[
\cM_\b = \left\{ S \subseteq E: \forall v \in V, |S \cap E_v| \le b_v \right\}.
\]
For $\lambda>0$ the Gibbs distribution $\mu = \mu_{G,\b,\lambda}$ is given by
\[
\mu(S) := \frac{\lambda^{|S|}}{Z}, \quad \forall S \in \cM_\b
\]
and the partition function $Z = Z_{G,\b}(\lambda)$ is defined as
\[
Z := \sum_{S \in \cM_\b} \lambda^{|S|}.
\]
Thus, for uniform $\b = b\bm{1}$ where $\bm{1}$ is the all-ones vector we obtain $b$-matchings, and for $b_v = d_v - b$ where $d_v$ is the degree of $v$ we get the complements of $b$-edge covers.

Our main contribution is to establish rapid mixing of the Glauber dynamics for sampling general $\b$-matchings for any $\b \in \bN^V$ on all bounded-degree graphs.
In each step of the Glauber dynamics, one picks an edge $e \in E$ uniformly at random and updates its status, $e \in S$ (occupied) or $e \notin S$ (unoccupied), conditional on the configuration of all other edges; in particular, if including $e$ violates the subset $S$ being a $\b$-matching then $e$ must be unoccupied in this update.
It is easy to show that the Glauber dynamics is ergodic for sampling $\b$-matchings.

\begin{theorem}[$\b$-Matchings]
\label{thm:b-matching-mixing}
Let $\Delta \ge 3$ be an integer and $G=(V,E)$ be an $n$-vertex graph of maximum degree $\Delta$.
Let $\b \in \bN^V$ be a vector of vertex thresholds. 
Then for any $\lambda > 0$, the Gibbs distribution $\mu = \mu_{G,\b,\lambda}$ over $\b$-matchings is $O_{\Delta,\lambda}(1)$-spectrally independent.
Furthermore, the Glauber dynamics for sampling from $\mu$ mixes in $O_{\Delta,\lambda}(n \log n)$ time.
\end{theorem}

We prove \cref{thm:b-matching-mixing} by the spectral independence method which was introduced recently in \cite{ALO20} and becomes a powerful tool for proving optimal mixing time of Glauber dynamics.
Our proof of spectral independence is inspired by \cite{chen2023near} and uses a recursive coupling to bound the Wasserstein $W_1$ distance under two distinct pinnings.
For uniformly random $b$-matchings with small $b$, our bound on spectral independence is $O(\Delta^b)$; see \cref{rmk:SI-b-matching} for more discussions.
We remark that one interesting open problem is to show spectral independence with a constant independent of $\Delta$ even just for the usual matchings (monomer-dimer model), since then one would obtain $O(n \log n)$ mixing of Glauber dynamics on all graphs even with unbounded degrees, using new powerful techniques such as the field dynamics \cite{CFYZ21,AJKPV22,CFYZ22,CE22}.

\subsection{Holant problem with log-concave signatures}
Both $b$-matchings and $b$-edge covers belong to a much more general family of models called \emph{Holant problems},
which can be understood as graphical models defined over subsets of edges of a given graph.
Examples and applications of Holant problems include also perfect matchings \cite{JSV04}, even subgraphs \cite{JS93,GJ18,LSS19,chen2022spectral,chen2023near,FGW22}, Fibonacci gates \cite{LWZ14}, spin systems on line graphs \cite{DHJM21,GLLZ21,BCR21,chen2022spectral}, etc.

We consider the following \emph{binary symmetric Holant problem}.
Let $G = (V, E)$ be a graph with $n$ vertices. 
For each vertex $v$ let $d_v$ denote the degree of $v$.
We consider a family of constraint functions on all vertices denoted by $\f = (f_v)_{v\in V}$,
where each vertex $v$ is associated with a constraint function $f_v: \bN \to \bR_{\ge 0}$.
Also, let $\blam = (\lambda_e)_{e \in E}\in \bR_{>0}^E$ be a vector of edge weights.
The Gibbs distribution $\mu = \mu_{G,\f,\blam}$ and the partition function $Z = Z_{G,\f,\blam}$ of the Holant problem is defined as
\begin{align*}
  \mu (\st) &:= \frac{1}{Z} \prod_{v\in V} f_v(|\st \cap E_v|) \prod_{e \in S} \lambda_e , \quad \forall \st \subseteq E; \\
  Z &:= \sum_{\st \subseteq E} \prod_{v\in V} f_v(|\st \cap E_v|) \prod_{e \in S} \lambda_e.
\end{align*}
When $f_v(k) = \mathbbm{1}\{k\le b_v\}$ for some $\b = (b_v)_{v\in V}$, the Holant problem becomes $\b$-matchings.

Holant problems can be defined more generally by allowing each $f_v: 2^{E_v} \to \bR_{\ge 0}$ to be a set function over subsets of neighboring edges of $v$.
In this paper we consider only the \emph{symmetric} case, i.e., the value of $f_v$ depends only on $|S\cap E_v|$, the number of adjacent edges in $S$.
Such symmetric constraint function $f_v$ can be equivalently identified by the sequence $f_v = [f_v(0),f_v(1),\dots,f_v(d_v)]$, which is called the \emph{signature} at $v$.

Our main result for Holant problems establishes spectral independence and rapid mixing of Glauber dynamics when all the signatures are \emph{log-concave} sequences.

\begin{definition}[Log-concave signature]
\label{def:logconcave}
A sequence $f = [f(0),f(1),\dots,f(d)]$ of non-negative real numbers is called a log-concave signature if it satisfies the following conditions:
\begin{enumerate}
  \item[(a)] Log-concavity: $f(k)^2 \ge f(k-1)f(k+1)$ for all $1\le k\le d-1$;
  \item[(b)] No internal zeros: if $f(k_1)>0$ and $f(k_2)>0$ for some $0\le k_1< k_2 \le d$, then $f(k)>0$ for all $k_1 \le k \le k_2$ (i.e., the support of $f$ is consecutive).
\end{enumerate}
\end{definition}

For example, the signature $f = [1,\dots,1,0,\dots,0]$ for the function $f(k) = \mathbbm{1}\{k\le b\}$ is log-concave.

\begin{theorem}[Holant problem, informal]
\label{thm:holant-mixing-imprecise}
Let $G=(V,E)$ be an $n$-vertex graph of maximum degree $\Delta$. Suppose that $\f = (f_v)_{v\in V}$ is a collection of log-concave signatures with $f_v(0) > 0$ for all $v \in V$.
Let $\blam \in \bR_{>0}^E$ be a vector of edge weights.
Then the Gibbs distribution $\mu = \mu_{G,\f,\blam}$ for the Holant problem $(G,\f,\blam)$ is $O_{\Delta,\f,\blam}(1)$-spectrally independent.
Furthermore, the Glauber dynamics for sampling from $\mu$ has modified log-Sobolev constant at least $1/(Cn)$ and mixing time at most $C n\log n$, where $C=C(\Delta,\f,\blam)$ does not depend on $n$.
\end{theorem}
This is informal because technically speaking, vectors $\f$ and $\blam$ are dependent on $n$ in dimensions. For a precise statement, see \cref{thm:holant-mixing}.

We remark that the Gibbs distribution $\mu = \mu_{G,\f,\blam}$ in \cref{thm:holant-mixing-imprecise} is supported on $\b$-matchings where $b_v = \max\{0\le k\le d_v: f_v(k) > 0\}$, and thus the Glauber dynamics for sampling from $\mu$ is ergodic.
Our assumptions of log-concave signatures in fact generalize previous works \cite{GLLZ21,BCR21,chen2022spectral} which essentially require that the generating polynomial $P(x) = \sum_{k=0}^d \binom{d}{k} f(k) x^k$ associated with every signature $f$ is real-rooted, which implies the log-concavity of $f$ by Newton inequalities, see e.g.~\cite{Bra15}.
Hence, \cref{thm:holant-mixing-imprecise} applies to many classes of Holant problems including the antiferromagnetic $2$-spin systems on line graphs.


\subsection{Hardcore model on claw-free graphs}
Another contribution of ours is that the Glauber dynamics has the optimal $O(n\log n)$ mixing time for the hardcore model on $n$-vertex bounded-degree claw-free graphs.
In the hardcore model, we are given a graph $G=(V, E)$ and $\blam = (\lambda_v)_{v\in V}\in \bR_{>0}^V$ a vector of vertex weights called fugacity.
A set $I\subseteq V$ is called an independent set if $e\not\subseteq I$ for all $e\in E$.
Let $\cI \subseteq 2^V$ be the set of all independent sets of $G$.
Define the Gibbs distribution $\mu=\mu_{G,\blam}$ and the partition function $Z=Z_{G,\blam}$, also called the (multivariate) independence polynomial, as
\begin{align*}
  \mu(I) &:= \frac 1Z \prod_{v\in I} \lambda_v, \qquad \forall I\in \cI;\\
  Z &:= \sum_{I\in \cI} \prod_{v\in I} \lambda_v.
\end{align*}
Specially, when $\blam=\lambda \bm{1}$, we denote the model as $\mu_{G,\lambda}$.

The Glauber dynamics is a natural Markov chain for sampling from the hardcore model.
In each step of the Glauber dynamics, a vertex $v\in V$ is picked uniformly at random, and its status, $v\in I$ (occupied) or $v\not \in I$ (unoccupied), is updated according to the configuration on all other vertices.
Specifically, if $v$ has at least one neighbor in the current independent set $I$, then nothing changes; if $v$ has no neighbors in $I$, then it becomes occupied with probability $\frac{\lambda_v}{1+\lambda_v}$ and unoccupied with probability $\frac{1}{1+\lambda_v}$.

We consider sampling from the hardcore model on a special class of graphs, the \emph{claw-free} graphs.
A graph $G=(V,E)$ is claw-free if it does not include an induced $K_{1,3}$.
In other words, there do not exist four distinct vertices $a,b,c,d\in V$ such that $(a,b),(a,c),(a,d)\in E$ but $(b,c),(b,d),(c,d)\not \in E$.
The class of claw-free graphs includes all line graphs by definition, and thus the hardcore model on claw-free graphs includes in particular the monomer-dimer model for matchings as a special case.

It was known that one can sample from the hardcore model on claw-free graphs in polynomial time.
Generalizing the approach from \cite{JS89,Jbook} for matchings, Matthews gave a Markov chain which mixes in $O(\Delta n^3)$ time where $\Delta$ is the maximum degree \cite{Mat08}.
Recently, \cite{DGM21} proved that the Glauber dynamics mixes in $O(n^5 \log n)$ time for claw-free graphs, and more generally in polynomial time for graphs with bounded bipartite pathwidth.

In another direction, Patel and Regts \cite{PR17} gave a polynomial-time deterministic algorithm (FPTAS) for approximating the partition function based on Barvinok's polynomial interpolation method \cite{Bar16book} and real-rootedness of the independence polynomial on claw-free graphs \cite{CS07,LR19}.
As is common for deterministic approximate counting algorithms, the exponent in $n$ in the running time depends on parameters of the model.

Our main result for the hardcore model on claw-free graphs is that the Glauber dynamics has optimal mixing when the maximum degree is bounded.
Again we prove optimal mixing by establishing spectral independence via a recursive coupling procedure.
\begin{theorem}[Hardcore model on claw-free graphs]
  \label{thm:hardcore-claw-free}
  Let $\Delta \ge 3$ be an integer and $G=(V,E)$ be an $n$-vertex claw-free graph of maximum degree $\Delta$.
  Let $\blam\in \bR_{>0}^V$ be a vector of fugacity
  with $\lambda_{\min} := \min_{v\in V} \lambda_v$ and $\lambda_{\max} := \max_{v\in V} \lambda_v$.
  Then the Gibbs distribution $\mu = \mu_{G,\blam}$ of the hardcore model is $2(1+\Delta\lambda_{\max})$-spectrally independent.
  Furthermore, the Glauber dynamics for sampling from $\mu$ has modified log-Sobolev constant at least $1/(Cn)$ and mixing time at most $C n\log n$, where $C=C(\Delta,\lambda_{\max},\lambda_{\min})$ does not depend on $n$.
\end{theorem}

\section{Preliminaries} \label{sec:prelim}
In this section we give definitions and lemmas that are needed.
We introduce with the Holant problems in mind, but the definitions and results work for the hardcore model with minor changes (e.g., replacing $E$ with $V$).

Let $2^E$ be the collection of all subsets of $E$.
We consider a distribution $\mu$ on $2^E$.
We view a subset $\st \subseteq E$ equivalently as a binary indicator vector $\sigma = \bm{1}_\st \in \{0,1\}^E$, where $\sigma_e=1$ for $e\in S$, and $\sigma_e=0$ for $e\not \in S$.

\begin{definition}[Pinning]
A pinning is a partial configuration $\tau \in \{0,1\}^\Lambda$ for some $\Lambda\subseteq E$ such that $\mu_\Lambda(\tau) > 0$, where $\mu_\Lambda$ is the marginal distribution on $\Lambda$.
Let $\mu^{\tau}$ denote the conditional distribution on $E \setminus \Lambda$.
\end{definition}

\begin{definition}[Marginal boundedness]
We say $\mu$ is $b$-marginally bounded if for all pinnings $\tau$ on $\Lambda \subseteq E$ and all $e \in E \setminus \Lambda$, we have either $b \le \mu^\tau(\sigma_e=1) \le 1-b$ or $\mu^\tau(\sigma_e=1) \in \{0,1\}$.
\end{definition}

\begin{definition}[Influence matrix]
Let $\mu$ be a distribution on $2^E$ and $\tau$ be a pinning on $\Lambda \subseteq E$.
The pairwise influence matrix $J_\mu^\tau \in \bR^{(E\backslash \Lambda) \times (E\backslash \Lambda)}$ is defined as following:
for all $e,f\in E\backslash \Lambda$, let
$$
J_\mu^\tau(e,f) =
\mu^\tau(\sigma_f = 1 | \sigma_e = 1) - \mu^\tau(\sigma_f = 1 | \sigma_e = 0)
$$
when $e\ne f$ and $\min\{\mu^\tau(\sigma_e=1), \mu^\tau(\sigma_e=0)\} >0$, and let $J_\mu^\tau(e,f)=0$ otherwise.
All eigenvalues of the influence matrix $J_\mu^\tau$ are real.
\end{definition}

\begin{definition}[Spectral independence \cite{ALO20}]
We say $\mu$ is $\eta$-spectrally independent if for all pinnings $\tau$ we have $\lambda_{\max}(J_\mu^\tau) \le \eta$.
\end{definition}

\begin{theorem}[\cite{chen2021optimal,BCCPSV22,chen2022spectral}] \label{thm:spec-indep-imply-mixing}
Let $\mu$ be the Gibbs distribution of a Holant problem on an $n$-vertex graph of maximum degree $\Delta$.
If $\mu$ is $\eta$-spectrally independent and $b$-marginally bounded, then the Glauber dynamics has modified log-Sobolev constant at least $1/(Cn)$ and mixing time at most $C n\log n$, where $C = C(\Delta,\eta,b)$ is a constant independent of $n$.
\end{theorem}

For two distributions $\nu,\pi$ on $2^E$, the $1$-Wasserstein distance between them is defined as
\[
W_1(\nu,\pi) = \inf_\mathcal{C} \bE_{(\sigma,\tau) \sim \mathcal{C}} \left[ d_{\mathrm{H}}(\sigma,\tau) \right],
\]
where the infimum is over all couplings between $\nu$ and $\pi$, and $d_{\mathrm{H}}(\cdot,\cdot)$ denotes the Hamming distance between two elements from $2^E$.
We use the following lemma from \cite{chen2023near} to establish spectral independence; see also \cite{CMM23,GGGH22,CLMM23} which use similar approach.

\begin{lemma}[\cite{chen2023near}] \label{lem:coup-indep-imply-spec-indep}
Let $\mu$ be the Gibbs distribution of a Holant problem.
Suppose that for some constant $\eta>0$, the following is true:
For any two pinnings $\tau$, $\tau'$ on the same subset $\Lambda\subseteq E$ which differ on exactly one edge, we have
\begin{align*}
  W_1(\mu^\tau, \mu^{\tau'}) \le \eta.
\end{align*}
Then $\mu$ is $\eta$-spectrally independent.
\end{lemma}

\section{Fast sampling for Holant problems with log-concave signatures}
We first give a precise statement for \cref{thm:holant-mixing-imprecise}.
It is helpful to define the following local generating polynomial associated with each vertex, as introduced in \cite{GLLZ21}.
\begin{definition}[Normalized generating polynomial]
For a signature $f = [f(0),f(1),\dots,f(d)]$ with $f(0) > 0$, define the \emph{normalized generating polynomial} to be
\[
P_f(x) = \frac{1}{f(0)}\sum_{k=0}^d \binom{d}{k} f(k) x^k.
\]
\end{definition}

\begin{definition}\label{def:para}
Let $(G,\f,\blam)$ be a Holant problem with log-concave signatures and $f_v(0)>0$ for all $v\in V$. We define:
\begin{align*}
    &r_{\max} := \max_{v\in V} \frac{f_v(1)}{f_v(0)}, \qquad
    &&r_{\min} := \min_{v\in V} \min_{k: f_v(k) > 0} \frac{f_v(k)}{f_v(k-1)},\\
    &\lambda_{\max} := \max_{e\in E} \lambda_e, \qquad
    &&\lambda_{\min} := \min_{e\in E} \lambda_e, \\
    &P_{\max} := \max_{v\in V} P_{f_v} \left( r_{\max}\lambda_{\max} \right), \qquad
    &&\Delta := \max_{v \in V} d_v.
\end{align*}
\end{definition}
\begin{theorem}[Holant problem]
  \label{thm:holant-mixing}
  Let $G=(V,E)$ be an $n$-vertex graph of maximum degree $\Delta$. Suppose that $\f = (f_v)_{v\in V}$ is a collection of log-concave signatures with $f_v(0) > 0$ for all $v \in V$.
  Let $\blam \in \bR_{>0}^E$ be a vector of edge weights.
  Then the Gibbs distribution $\mu = \mu_{G,\f,\blam}$ for the Holant problem $(G,\f,\blam)$ is $O(P_{\max})$-spectrally independent, where $P_{\max}$ is defined in \cref{def:para}.
  Furthermore, the Glauber dynamics for sampling from $\mu$ has modified log-Sobolev constant at least $1/(Cn)$ and mixing time at most $C n\log n$ for some $C=C(\Delta,P_{\max},r_{\min},\lambda_{\min})$, where $P_{\max},r_{\min},\lambda_{\min}$ are defined in \cref{def:para}.
\end{theorem}
See also \cref{lem:P-max-bound} for a simple upper bound of $P_{\max}$ in terms of $\Delta,r_{\max},\lambda_{\max}$.

By \cref{thm:spec-indep-imply-mixing}, to prove \cref{thm:holant-mixing}, it suffices to establish spectral independence and marginal boundedness.
We focus on spectral independence in \cref{subsec:SI} whose proof is based on the log-concavity of signatures.
We give the marginal bound analysis in \cref{subsec:marginal-bounds}. The proofs of main results can be found in \cref{subsec:proof-main}.

\subsection{Spectral independence}
\label{subsec:SI}

In this subsection we derive a constant bound of spectral independence via \cref{lem:coup-indep-imply-spec-indep}.

\begin{proposition}[Spectral independence]
\label{lem:SI}
Under the assumptions of \cref{thm:holant-mixing}, the Gibbs distribution $\mu$ of the Holant problem $(G,\f,\blam)$ is $O(P_{\max})$-spectrally independent.
\end{proposition}

For spectral independence we need to consider the conditional distribution $\mu^\tau$ under an arbitrary pinning $\tau$.
We note that a pinning $\tau$ on a subset $\Lambda \subseteq E$ of edges induces a Holant problem on the subgraph $G \setminus \Lambda$.
To formalize this relationship, it is helpful to define the following notation of downward shifting operator.
\begin{definition}[Downward shifting]
  For a function $f: \bN \to \bR_{\ge 0}$, we define the function $\cD f: \bN \to \bR_{\ge 0}$ as
  $$(\cD f)(k) = f(k+1), \quad \forall k \in \bN. $$
  We further define $\cD^m f = \cD(\cD^{m-1} f)$ for integer $m \ge 1$.
\end{definition}

Let $(G,\f,\blam)$ be a Holant problem.
For a subset $U\subseteq V$ of vertices, we define $\cD_U \f$ as
\begin{align*}
  (\cD_U \f)_v = \left\{
    \begin{array}{cc}
      \cD f_v, & v\in U,\\
      f_v,  & v\not \in U.
    \end{array}
  \right.
\end{align*}
For a pinning $\tau$, we define $\cD_\tau \f$ as for all $v \in V$,
\[
(\cD_\tau \f)_v = \cD^{|\tau \cap E_v|} f_v.
\]
Observe that if an adjacent edge of a vertex $v$ is pinned to be occupied, it corresponds to changing the signature of $v$ from $f_v$ to $\cD f_v$.
Hence, the Holant problem $(G,\f,\blam)$ with pinning $\tau$ on $\Lambda \subseteq E$ induces a smaller instance of Holant problem $(G\setminus \Lambda, \cD_\tau \f, \blam_{E\setminus \Lambda})$ on the subgraph $G\setminus \Lambda$.

\begin{observation}\label{obs:pinning}
Consider a Holant problem $(G,\f,\blam)$ satisfying the condition in \cref{thm:holant-mixing}.
Then for any pinning the induced Holant problem also satisfies the conditions.
Furthermore, all parameters in \cref{def:para} are ``monotone'' in pinnings;
i.e., we have that $r_{\max},\lambda_{\max},P_{\max},\Delta$ are non-increasing under any pinning and $r_{\min},\lambda_{\min}$ are non-decreasing.
\end{observation}

\begin{proof}
The observation follows from that all signatures are log-concave and hence for all $v \in V$,
\[
\frac{f_v(k)}{f_v(0)} \ge \frac{f_v(k+\ell)}{f_v(\ell)}
\]
where $0 \le k \le d_v$ and $0\le \ell \le d_v-k$ (assuming $0/0 = 0$).
\end{proof}

By \cref{obs:pinning}, $\mu^\tau$ corresponds to an induced Holant problem still satisfying the conditions of \cref{thm:holant-mixing}; hence it suffices to focus on the no-pinning case.
The following proposition gives the key step for bounding the spectral independence constant via \cref{lem:coup-indep-imply-spec-indep}.
It upper bounds the expected number of discrepancies when one signature $f_v$ is changed to $\cD f_v$, i.e., the difference between an adjacent (half-)edge of $v$ is occupied and unoccupied.

\begin{proposition} \label{prop:holant-coupling}
Let $(G,f,\lambda)$ be a Holant problem satisfying the conditions in \cref{thm:holant-mixing} with Gibbs distribution $\mu = \mu_{G,\f,\blam}$.
Suppose $v\in V$ is a vertex with $f_v(1)>0$, and let $\mu' = \mu_{G,\cD_v\f,\blam}$ be the Gibbs distribution of the Holant problem obtained by changing $f_v$ to $\cD f_v$.
Then we have
\begin{align*}
W_1(\mu, \mu')
\le P_{\max} -1.
\end{align*}
\end{proposition}

\begin{algorithm}[t]
\caption{Coupling procedure for Holant problems}
\label{alg:couple-holant}
\begin{algorithmic}[1]
\Procedure{Couple}{$G, \f, \blam, v$}
\State \textbf{Input:} $(G,\f,\blam)$ a Holant problem, $v\in V$ a disagreeing vertex
\State \textbf{Output:} A pair of random configurations $(\sigma,\sigma') \in 2^E \times 2^E$ drawn from a coupling between $\mu = \mu_{G,\f,\blam}$ and $\mu' = \mu_{G,\cD_v\f,\blam}$

\medskip
\If{$v$ is isolated}
\State Sample $\sigma\sim \mu$
\State \Return $(\sigma,\sigma)$
\EndIf
\State Choose $e=\{u,v\}\in E$ such that
$\mu_e(1) \ge \mu'_e(1)$.
\Comment{\cref{lem:log-concave-monotone}}
\label{line:alg-couple-holant:choose-edge}
\State Sample $(\sigma_e,\sigma'_e)$ from an optimal coupling of $(\mu_e,\mu'_e)$. \label{line:alg-couple-holant:couple-edge}
\If{$\sigma_e = \sigma'_e= 0$}
\label{line:alg-couple-holant:coupled-no-edge}
\State $(\sigma_{E\backslash e}, \sigma'_{E\backslash e}) \gets \textsc{Couple}(G\backslash e, \f, \blam_{E\setminus e}, v)$
\ElsIf{$\sigma_e = \sigma'_e=1$}
\label{line:alg-couple-holant:coupled-edge}
\State $(\sigma_{E\backslash e}, \sigma'_{E\backslash e}) \gets \textsc{Couple}(G\backslash e, \cD_e\f, \blam_{E\setminus e}, v)$
\Else\Comment{We must have $\sigma_e = 1$, $\sigma'_e=0$}
\label{line:alg-couple-holant:uncoupled-edge}
\State $(\sigma'_{E\backslash e}, \sigma_{E\backslash e}) \gets \textsc{Couple}(G\backslash e, \cD_v \f, \blam_{E\setminus e}, u)$
\EndIf
\State \Return $(\sigma_e \cup \sigma_{E\backslash e}, \sigma'_e \cup \sigma'_{E\backslash e})$
\EndProcedure
\end{algorithmic}
\end{algorithm}

Our coupling between $\mu$ and $\mu'$ is inspired by \cite{chen2023near} which proves spectral independence for weighted even subgraphs with signatures $[1,a,1,a,\dots]$ for some $a>0$. Note that such signatures have period two which is crucial for the coupling arguments in \cite{chen2023near}. Our new ingredient is to construct a coupling without periodicity of signatures but incorporating the log-concavity in a suitable way.

\begin{proof}[Proof of \cref{prop:holant-coupling}]
  We construct a simple coupling between $\mu = \mu_{G,\f,\blam}$ and $\mu' = \mu_{G,\cD_v \f,\blam}$ using \cref{alg:couple-holant}.

  \paragraph{Coupling.} We prove that \cref{alg:couple-holant} produces a coupling between $\mu$ and $\mu'$.
  We prove this by induction on the number of edges of $G$.

  Base case is when $v$ is isolated. In this case we have $\mu =\mu'$. The algorithm produces the identity coupling.

  Induction step: Suppose $v$ is not isolated. By \cref{lem:log-concave-monotone}, there exists an edge $e\in E_v$ such that $\mu_e(1) \ge \mu'_e(1)$, so \cref{line:alg-couple-holant:choose-edge} runs successfully; we shall prove \cref{lem:log-concave-monotone} right after this proof.
  By \cref{line:alg-couple-holant:couple-edge}, the marginal distributions on edge $e$ are correct.
  In particular, note that because in \cref{line:alg-couple-holant:couple-edge} we choose an optimal coupling, it is impossible to have $\sigma_e=0$ and $\sigma'_e=1$ since $\mu_e(1) \ge \mu'_e(1)$.
  We need to prove that the recursive calls produce the desired distributions on $E\backslash e$.

  \begin{itemize}
    \item Case 1: $\sigma_e=\sigma'_e=0$.
    By induction hypothesis,
    \begin{align*}
      \sigma_{E\backslash e} &\sim \mu_{G\backslash e, \f,\blam_{E\setminus e}} = \mu_{E\backslash e}(\cdot \mid \sigma_e=0), \\
      \sigma'_{E\backslash e} &\sim \mu_{G\backslash e, \cD_v \f,\blam_{E\setminus e}} = \mu'_{E\backslash e}(\cdot \mid \sigma'_e=0).
    \end{align*}
    \item Case 2: $\sigma_e=\sigma'_e=1$.
    By induction hypothesis,
    \begin{align*}
      \sigma_{E\backslash e} &\sim \mu_{G\backslash e, \cD_e \f,\blam_{E\setminus e}} = \mu_{E\backslash e}(\cdot \mid \sigma_e=1), \\
      \sigma'_{E\backslash e} &\sim \mu_{G\backslash e, \cD_v \cD_e f,\blam_{E\setminus e}} = \mu'_{E\backslash e}(\cdot \mid \sigma'_e=1).
    \end{align*}
    \item Case 3: $\sigma_e=1$, $\sigma'_e=0$.
    By induction hypothesis,
    \begin{align*}
      \sigma_{E\backslash e} &\sim \mu_{G\backslash e, \cD_u \cD_v \f,\blam_{E\setminus e}} = \mu_{E\backslash e}(\cdot \mid \sigma_e=1), \\
      \sigma'_{E\backslash e} &\sim \mu_{G\backslash e, \cD_v \f,\blam_{E\setminus e}} = \mu'_{E\backslash e}(\cdot \mid \sigma'_e=0).
    \end{align*}
  \end{itemize}

  In all three cases, we see that $\sigma_{E\backslash e}$ and $\sigma'_{E\backslash e}$ have the desired distributions.
  So the algorithm returns correctly.

  \paragraph{$W_1$ distance.}
  Let us bound the expected $\ell_1$ distance under the coupling generated by \cref{alg:couple-holant}.
  We prove by induction on the number of edges that the expected $\ell_1$ distance is at most
  $
  P_{\max}  -1,
  $ as defined in \cref{def:para}.

  Base case is when $v$ is isolated. In this case the $\ell_1$ distance is $0$.

  Induction step: Suppose that $v$ is not isolated.
  We consider recursive calls of \cref{alg:couple-holant} until: (1) the input vertex becomes some other vertex $u \neq v$, or (2) the algorithm halts.
  We claim that Case~(2) happens with probability at least $\mu_{E_v}(\bm{0})$, which is the probability that all edges in $E_v$ are unoccupied under the Gibbs distribution $\mu$.

  To see this, let $\mathcal{A}$ be the following event: For all $1\le i \le d$ where $d = d_v = |E_v|$, in the $i$-th call of \cref{alg:couple-holant}, the algorithm picks an edge from $E_v$, denoted by $e_i$, which has not yet be chosen, and sets its values to be $\sigma_{e_i} = \sigma'_{e_i} = 0$ in both samples. Note that if $\mathcal{A}$ occurs then the algorithm halts in the $(d+1)$-th call and the $\ell_1$ distance is $0$.
  For simplicity of the proof we assume that there is a total ordering of edges and in \cref{line:alg-couple-holant:choose-edge} we always pick the smallest edge satisfying the requirement.
  Hence, the sequence of edges $(e_1,\dots,e_d)$ associated with the event $\mathcal{A}$ is fixed and deterministic: $e_i$ is the smallest edge in $E_v \setminus \{e_1,\dots,e_{i-1} \}$ such that
  \begin{align}\label{eq:choice-of-ei}
  \mu(\sigma_e = 1 \mid \sigma_{e_1} = \cdots = \sigma_{e_{i-1}} = 0)
  \ge \mu'(\sigma'_e = 1 \mid \sigma'_{e_1} = \cdots = \sigma'_{e_{i-1}} = 0).
  \end{align}
  It follows that
  \begin{align*}
    \Pr\left( \text{Case (2)} \right)
    \ge \Pr\left( \mathcal{A} \right)
    &= \prod_{i=1}^d \Pr\left( \sigma_{e_i} = \sigma'_{e_i} = 0 \,\Big\vert\, \forall j < i, \sigma_{e_j} = \sigma'_{e_j} = 0 \right) \\
    &\stackrel{(*)}{=} \prod_{i=1}^d \mu(\sigma_e = 0 \mid \sigma_{e_1} = \cdots = \sigma_{e_{i-1}} = 0)  \\
    &= \mu(\sigma_{E_v} = \bm{0}),
  \end{align*}
  where $(*)$ is because of \cref{eq:choice-of-ei} and the optimal coupling of $(\sigma_e,\sigma'_e)$.

  To summarize, Case~(2) happens with probability at least $\mu_{E_v}(\bm{0})$ and the $\ell_1$ distance is $0$ in this case;
  Case~(1) happens with probability at most $1-\mu_{E_v}(\bm{0})$ and the expected $\ell_1$ distance is at most $1 + (P_{\max} - 1) = P_{\max}$, where $1$ comes from the edge $\{u,v\}$ and $P_{\max} - 1$ comes from induction hypothesis.
  Therefore, the expected $\ell_1$ distance produced by the root call is at most
  \begin{align}\label{eq:lemma-all0}
    \left( 1-\mu_{E_v}(\bm{0}) \right) P_{\max}
    \le P_{\max} - 1,
  \end{align}
  where the inequality follows from $\mu_{E_v}(\bm{0}) \ge 1/P_{\max}$ by \cref{lem:all-zero-prob}, whose proof is technical and postponed to \cref{subsec:marginal-bounds}.
  This completes the induction.
\end{proof}

Note that as long as it holds $\mu_{E_v}(\bm{0}) = \Omega(1)$ one can deduce $O(1)$ Wasserstein distance using the inductive argument in \cref{prop:holant-coupling}.

We now present and prove \cref{lem:log-concave-monotone} which is crucially used in \cref{line:alg-couple-holant:choose-edge} of \cref{alg:couple-holant}.

\begin{lemma} \label{lem:log-concave-monotone}
Let $(G,\f,\blam)$ be a Holant problem and $v \in V$ be a vertex such that $\mu=\mu_{G,\f,\blam}$ and $\mu'=\mu_{G,\cD_v\f,\blam}$ are both well-defined.
If $f_v$ is a log-concave signature, then we have
\begin{align} \label{eqn:lem-log-concave-monotone}
  \bE_{\sigma\sim\mu} |\sigma\cap E_v| \ge \bE_{\sigma'\sim\mu'} |\sigma'\cap E_v|.
\end{align}
In particular, there exists $e\in E_v$ such that
\begin{align} \label{eq:exists-e}
  \mu(\sigma_e=1) \ge \mu'(\sigma'_e=1).
\end{align}
\end{lemma}
\begin{proof}
  Let $d_v = |E_v|$ be the degree of $v$ in $G$.
  For $0\le k \le d_v$, define
  \begin{align*}
    \Phi_k := \sum_{\sigma\in 2^E:\, |\sigma \cap E_v| = k} \prod_{u\in V\backslash v} f_u(|\sigma\cap E_u|) \prod_{e\in E:\, \sigma_e=1} \lambda_e.
  \end{align*}
  Then we have that
  \begin{align*}
    \bE_{\sigma\sim\mu} |\sigma\cap E_v| &= \frac{\sum_{k=0}^{d_v} k f_v(k) \Phi_k}{\sum_{k=0}^{d_v} f_v(k) \Phi_k}, \\
    \bE_{\sigma'\sim\mu'} |\sigma'\cap E_v| &= \frac{\sum_{k=0}^{d_v} k f_v(k+1) \Phi_k}{\sum_{k=0}^{d_v} f_v(k+1) \Phi_k}.
  \end{align*}
  It follows that
  \begin{align*}
    &\left(\sum_{k=0}^{d_v} k f_v(k) \Phi_k\right) \left(\sum_{\ell=0}^{d_v} f_v(\ell+1) \Phi_\ell\right)
    - \left(\sum_{k=0}^{d_v} f_v(k) \Phi_k\right) \left(\sum_{\ell=0}^{d_v} \ell f_v(\ell+1) \Phi_\ell \right) \\
    ={}& \sum_{k,\ell=0}^{d_v} (k-\ell) f_v(k) f_v(\ell+1) \Phi_k\Phi_\ell  \\
    ={}& \frac 12 \sum_{k,\ell=0}^{d_v} (k-\ell) \big( f_v(k) f_v(\ell+1) - f_v(k+1) f_v(\ell) \big) \Phi_k \Phi_\ell
    \ge 0,
  \end{align*}
  where the last step is because
  of the log-concavity assumption of $f_v$.
  Therefore~\cref{eqn:lem-log-concave-monotone} holds as desired, and \cref{eq:exists-e} is an immediate consequence of \cref{eqn:lem-log-concave-monotone}.
\end{proof}

We are now ready to prove \cref{lem:SI} on spectral independence.

\begin{proof}[Proof of \cref{lem:SI}]
We prove spectral independence via \cref{lem:coup-indep-imply-spec-indep}.
Let $\tau$ be a pinning on $\Lambda \subseteq E$ and let $e = \{u,v\} \in E \setminus \Lambda$ such that $0<\mu^\tau_e(1)<1$.
For two pinnings $\tau' = \tau \cup \{e \gets 0\}$, $\tau'' = \tau \cup \{e \gets 1\}$ on $\Lambda \cup \{e\}$ which differ only at $e$, we consider the two conditional distributions
\begin{align*}
 \mu^{\tau'} &= \mu_{G,\f,\blam}(\cdot \mid \sigma_\Lambda = \tau, \sigma_e = 0) = \mu_{\wt{G},\wt{\f},\wt{\blam}}, \\
 \mu^{\tau''} &= \mu_{G,\f,\blam}(\cdot \mid \sigma_\Lambda = \tau, \sigma_e = 1) = \mu_{\wt{G},\cD_v \cD_u \wt{\f},\wt{\blam}},
\end{align*}
where $\wt{G} = G \setminus \Lambda \setminus \{e\}$, $\wt{\f} = \cD_\tau \f$, and $\wt{\blam} = \blam_{E \setminus \Lambda \setminus\{e\}}$.
By \cref{prop:holant-coupling} and the triangle inequality, we deduce that
  \begin{align*}
      W_1( \mu^{\tau'}, \mu^{\tau''} )
      \le W_1\left( \mu_{\wt{G},\wt{\f},\wt{\blam}}, \mu_{\wt{G},\cD_u \wt{\f},\wt{\blam}} \right)
      + W_1\left( \mu_{\wt{G},\cD_u \wt{\f},\wt{\blam}}, \mu_{\wt{G},\cD_v \cD_u \wt{\f},\wt{\blam}} \right)
      \le 2 (P_{\max}-1).
  \end{align*}
Note that all these Holant sub-problems under downward shifting operators still satisfy the assumptions of \cref{thm:holant-mixing} by \cref{obs:pinning}.
Therefore, by \cref{lem:coup-indep-imply-spec-indep} $\mu_{G,\f,\blam}$ is $2(P_{\max}-1)$-spectrally independent.
\end{proof}

\subsection{Marginal bounds}
\label{subsec:marginal-bounds}

We first lower bound the probability that all adjacent edges of a vertex are unoccupied, which justifies \cref{eq:lemma-all0}.
\begin{lemma} \label{lem:all-zero-prob}
Let $(G,\f,\blam)$ be a Holant problem satisfying the assumptions of \cref{thm:holant-mixing}.
Then for any vertex $v$, we have
\begin{align*}
  \mu_{E_v}(\bm{0}) \ge \frac{1}{P_v(r_{\max} \lambda_{\max})}.
\end{align*}
\end{lemma}

\begin{proof}
For any $\sigma \in 2^{E \setminus E_v}$ and $\tau \in 2^{E_v}$ viewed as subsets,
suppose $\tau = \{e_1,\dots,e_k\} \subseteq E_v$ where $k=|\tau|$ and $e_i=\{u_i,v\}$ for each $i$,
and we have that
\begin{align*}
\frac{\mu(\sigma \cup \tau)}{\mu(\sigma)}
&\le
\frac{f_v(k)}{f_v(0)}
\left( \prod_{i=1}^k \max_{1\le \ell \le d_{u_i}} \frac{f_{u_i}(\ell)}{f_{u_i}(\ell-1)} \right)
\left( \prod_{i=1}^k \lambda_{e_i} \right) \\
&=
\frac{f_v(k)}{f_v(0)}
\prod_{i=1}^k \left(\frac{f_{u_i}(1)}{f_{u_i}(0)} \lambda_{e_i}\right) \\
&\le \frac{f_v(k)}{f_v(0)}(r_{\max} \lambda_{\max})^{k}.
\end{align*}
Summing over $\sigma$ and $\tau$, we obtain that
\begin{align*}
1 &= \sum_{\sigma \in 2^{E \setminus E_v}} \sum_{\tau \in 2^{E_v}} \mu(\sigma \cup \tau) \\
&\le \sum_{\sigma \in 2^{E \setminus E_v}} \sum_{\tau \in 2^{E_v}} \frac{f_v(|\tau|)}{f_v(0)}(r_{\max} \lambda_{\max})^{|\tau|} \mu(\sigma) \\
&= \sum_{\sigma \in 2^{E \setminus E_v}} \mu(\sigma) \sum_{\tau \in 2^{E_v}} \frac{f_v(|\tau|)}{f_v(0)}(r_{\max} \lambda_{\max})^{|\tau|} \\
&\le \mu_{E_v}(\bm{0}) P_v(r_{\max} \lambda_{\max}),
\end{align*}
as claimed.
\end{proof}

We now give marginal lower bounds.
\begin{proposition}[Marginal boundedness]
\label{prop:marginal-bound}
  Let $(G,\f,\blam)$ be a Holant problem satisfying the assumptions of \cref{thm:holant-mixing}.
  Then for any pinning $\tau$ on a subset $\Lambda \subseteq E$ and for any edge $e \in E \setminus \Lambda$ that can be occupied, we have
  \begin{align*}
    \mu^\tau(\sigma_e=0) \ge \frac{1}{P_{\max}}
    \quad\text{and}\quad
    \mu^\tau(\sigma_e=1) \ge \frac{r_{\min}^2 \lambda_{\min}}{P_{\max}^2}.
  \end{align*}
\end{proposition}
\begin{proof}
By the monotonicity of pinning given in \cref{obs:pinning} it suffices to focus on the no-pinning case.
  Let $e=\{u,v\}$.
  The first part follows from \cref{lem:all-zero-prob}. Let us prove the second part.

  By applying \cref{lem:all-zero-prob} twice, we have
  \begin{align*}
    \mu(\sigma_{E_u\cup E_v} = \bm{0}) = \mu(\sigma_{E_u} = \bm{0}) \cdot \mu(\sigma_{E_v \setminus E_u} = \bm{0} \mid \sigma_{E_u} = \bm{0})
    \ge \frac{1}{P_{\max}^2},
  \end{align*}
  where we note that $\mu(\cdot \mid  \sigma_{E_u} = \bm{0})$ corresponds to an induced Holant problem whose $P_{\max}$ is no bigger by \cref{obs:pinning}.
  Then we have
  \begin{align*}
    \mu(\sigma_e = 1) & \ge \mu(\sigma_e=1, \sigma_{E_u\cup E_v \setminus \{e\}}=\bm{0}) \\
    & = \frac{f_u(1)f_v(1)}{f_u(0)f_v(0)} \lambda_e \cdot \mu(\sigma_{E_u\cup E_v}=\bm{0}) \\
    &\ge \frac{r_{\min}^2 \lambda_{\min}}{P_{\max}^2}.
  \end{align*}
  This verifies the second part.
\end{proof}

The following lemma provides an upper bound on $P_{\max}$ which in turn gives simple and clean constant bounds on spectral independence and marginal boundedness.
\begin{lemma} \label{lem:P-max-bound}
For a Holant problem $(G,\f,\blam)$ with log-concave signatures and $f_v(0)>0$ for all $v\in V$, we have
\[
P_{\max} \le \left( r_{\max}^{2} \lambda_{\max} + 1 \right)^\Delta.
\]
\end{lemma}

\begin{proof}
We note that for any vertex $v$,
\[
\frac{f_v(k)}{f_v(0)} = \prod_{i=1}^k \frac{f_v(i)}{f_v(i-1)} \le \left( \frac{f_v(1)}{f_v(0)} \right)^k \le r_{\max}^k
\]
by the log-concavity of $f_v$.
Hence, we deduce that
\[
P_{f_v}(x) = \frac{1}{f_v(0)}\sum_{k=0}^{d_v} \binom{d_v}{k} f_v(k) x^k
\le \sum_{k=0}^{d_v} \binom{d_v}{k} r_{\max}^k x^k
= \left( 1+r_{\max}x \right)^{d_v}.
\]
Letting $x = r_{\max}\lambda_{\max}$ and taking maximum over $v$ gives the desired bound.
\end{proof}

\subsection{Proofs of main results}
\label{subsec:proof-main}

\begin{proof}[Proof of \cref{thm:holant-mixing}]
  By \cref{lem:SI}, $\mu$ is $O(P_{\max})$-spectrally independent.
  By \cref{prop:marginal-bound}, $\mu$ is $b$-marginally bounded for $b = b(P_{\max},r_{\min},\lambda_{\min})>0$.
  Using \cref{thm:spec-indep-imply-mixing}, we finish the proof.
\end{proof}

\begin{proof}[Proof of \cref{thm:b-matching-mixing}]
  Because $[1,\dots,1,0,\dots,0]$ is a log-concave signature, we can apply \cref{thm:holant-mixing}.
  For the $b$-matching problem $(G,\b,\lambda)$, we have $r_{\max}=r_{\min}=1$ and $\lambda_{\max}=\lambda_{\min}=\lambda$.
  Therefore $P_{\max}$ is upper bounded by a function of $\Delta$ and $\lambda$.
  So \cref{thm:b-matching-mixing} is a corollary of \cref{thm:holant-mixing}.
\end{proof}

\begin{remark}\label{rmk:SI-b-matching}
Note that for uniformly random $b$-matchings on a graph of maximum degree $\Delta$ where $1\le b < \Delta$, we have $r_{\max} = r_{\min} = 1$, $\lambda_{\max} = \lambda_{\min} = 1$, and
\[
P_{\max} \le \sum_{k=0}^b \binom{\Delta}{k}.
\]
Hence, for small $b \ll \Delta$ we have $O(\Delta^b)$-spectral independence, and for all $b$ we have $O(2^\Delta)$-spectral independence.
\end{remark}

\subsection{Further discussions}

\cref{thm:holant-mixing} establishes spectral independence for a broad class of Holant problems with log-concave signatures.
The support of the corresponding Gibbs distribution is over all $\b$-matchings for some $\b \in \bN^V$.
Note that $O(1)$-spectral independence fails for the signature $[0,1,0,\dots,0]$ which corresponds to perfect matchings.
Moreover, it also fails for the signature $[0,1,\dots,1,0]$ in the $\ell_\infty$ sense (maximum absolute row sum of influence matrices), as illustrated by the examples below.

\begin{itemize}
\item Consider a path $v_0 \leftrightarrow v_1 \leftrightarrow \dots \leftrightarrow v_n$.
Suppose the signature at $v_{2i-1}$ is $[1,1,0]$ for all $i \ge 1$, i.e., requiring at least one adjacent edge unoccupied, and the signature at $v_{2i}$ is $[0,1,1]$ for all $i \ge 1$, i.e., requiring at least one adjacent edge occupied.
If $\sigma(v_0v_1)=1$, then there is only one feasible configuration where $\sigma(v_{2i-1}v_{2i}) = 0$ and $\sigma(v_{2i}v_{2i+1}) = 1$.
If $\sigma(v_0v_1)=0$, then one can check that there are $n$ feasible configurations.
The absolute sum of influences of edge $v_0v_1$ on all other edges is $\Omega(n)$.

\item Consider a path $v_0 \leftrightarrow v_1 \leftrightarrow \dots \leftrightarrow v_n$ together with edges $u_iv_i$ where $1\le i< n$.
Suppose every vertex $v_i$ has signature $[0,1,1,0]$ for $i\ge 1$.
Fix $\sigma(u_{2i-1}v_{2i-1}) = 1$ and $\sigma(u_{2i}v_{2i}) = 0$ for all $i\ge 1$ as a pinning,
then the resulted Holant problem is equivalent to the previous example.
 So the absolute sum of influences of $v_0v_1$ is $\Omega(n)$.
\end{itemize}


\section{Hardcore model on claw-free graphs} \label{sec:hardcore-claw-free}
In this section we prove \cref{thm:hardcore-claw-free} from introduction. The proof idea is similar to \cref{thm:holant-mixing} on a high level. We establish an upper bound on the Wasserstein $W_1$ distance between distributions under different pinnings, which implies spectral independence and fast mixing via \cref{lem:coup-indep-imply-spec-indep} and \cref{thm:spec-indep-imply-mixing}. The $W_1$ upper bound is proved using a recursive coupling (\cref{alg:couple-hardcore-claw-free}), which shares some similarities with \cref{alg:couple-holant} but requires new ideas to utilize claw-freeness of the graph.

While the hardcore model is a distribution on $2^V$ rather than $2^E$, notations in \cref{sec:prelim} can be applied here with simple changes. For example, we view a subset $I\subseteq V$ equivalently as a binary indicator vector $\sigma = \bm{1}_I \in \{0,1\}^V$. We similarly define pinnings, marginal boundedness, influence matrices, and spectral independence for the hardcore model; importantly, \cref{thm:spec-indep-imply-mixing} and \cref{lem:coup-indep-imply-spec-indep} still hold.
We refer to \cite{ALO20,chen2021optimal} for formal definitions and precise statements.

\subsection{Preliminaries on claw-free graphs}

A graph is called \emph{claw-free} if it does not contain $K_{1,3}$ (that is, a star graph comprising three edges, three leaves, and a central vertex) as an induced subgraph.
Note that any induced subgraph of a claw-free graph is still claw-free by definition.
The class of claw-free graphs includes in particular all line graphs.

For a vertex $v$ of a graph $G=(V,E)$,
let $N_v = \{u \in V: (u,v) \in E\}$ denote the neighborhood of $v$, and let $N^*_v = v \cup N_v$ denote the closed neighborhood including $v$ itself.
We say a vertex $v$ is \emph{simplicial} if its neighborhood $N_v$ (or, equivalently, its closed neighborhood $N^*_v$) forms a clique, i.e., every two neighbors of $v$ are adjacent.
The following simple lemma is very helpful to us.
\begin{claim}\label{lem:claw-free-fact}
Suppose $G = (V,E)$ is a claw-free graph and $v \in V$ is a vertex.
Let $u \in N_v$ be a neighbor of $v$.
Then the subgraph $G \setminus (N^*_v \setminus u)$ is claw-free and $u$ is a simplicial vertex of $G \setminus (N^*_v \setminus u)$.
\end{claim}

\begin{proof}
The subgraph $G \setminus (N^*_v \setminus u)$ is claw-free since it is an induced subgraph on $V \setminus (N^*_v \setminus u)$.
Suppose for contradiction that $u$ is not simplicial in $G \setminus (N^*_v \setminus u)$.
By definition there exist two neighbors $w_1,w_2$ of $u$ in $G \setminus (N^*_v \setminus u)$ that are not adjacent.
Then, $\{u,v,w_1,w_2\}$ forms a claw centered at $u$ in $G$: $u$ is adjacent to all of $v,w_1,w_2$ by our choice, and $w_1,w_2$ are not adjacent to $v$ since $w_1,w_2 \notin N^*_v$. This is a contradiction and hence $u$ must be simplicial.
\end{proof}

\subsection{Spectral independence}

We establish spectral independence in this subsection.
For a vertex $v$ and a spin $i \in \{0,1\}$, we use the notation $v \gets i$ to represent the pinning $\sigma_v = i$.
Hence, if $\Lambda \subseteq V$ is a subset of vertices, $\tau$ is a pinning on $\Lambda$, and $v \in V \setminus \Lambda$, then $\mu^{\tau,v\gets i}$ represents the conditional Gibbs distribution on the subset $V \setminus \Lambda \setminus v$ conditioned on both $\sigma_\Lambda = \tau$ and $\sigma_v = i$.

\begin{proposition}[Spectral independence] \label{prop:hardcore-claw-free-spec-indep}
  Work under the setting of \cref{thm:hardcore-claw-free}.
  For any pinning $\tau$ on a subset $\Lambda\subseteq V$ and any vertex $v \in V \setminus \Lambda$, we have
  \begin{align*}
    W_1\left( \mu^{\tau,v\gets 0}, \mu^{\tau,v\gets 1} \right) \le 2(1+\Delta\lambda_{\max}).
  \end{align*}
  In particular, $\mu$ is $2(1+\Delta\lambda_{\max})$-spectrally independent.
\end{proposition}

As a standard trick for the hardcore model, we can view the the conditional Gibbs distribution under a pinning $\tau$ on $\Lambda \subseteq V$ as a hardcore model on an induced subgraph by removing all pinned vertices in $\Lambda$ together with all neighbors of those fixed to be occupied (pinned to spin $1$).
Note that the maximum degree of the resulting subgraph does not increase.
Thus, for simplicity we consider the case without pinnings except at $v$.

We first argue that for claw-free graphs, to prove \cref{prop:hardcore-claw-free-spec-indep} it suffices to consider only when $v$ is a simplicial vertex.
Sample $\xi \sim \mu_{N_v}^{v\gets 0}$ and $\xi' \sim \mu_{N_v}^{v\gets 1}$ two configurations on the neighborhood of $v$.
Observe that we must have $\xi'=\bm{0}$ since $v$ is occupied.
Suppose $\|\xi\|_1 = m$ and denote the occupied vertices in $\xi$ by $\{u_1,\dots,u_m\}$ under any ordering. We must have $m\le 2$; otherwise, since there can be no edge between any pair of occupied vertices, the set $\{v,u_1,u_2,u_3\}$ forms a claw.
For $0\le i \le m$, let $\xi^i$ be the configuration on $N_v$ with $\{u_1,\dots,u_i\}$ occupied and all other vertices in $N_v$ unoccupied, so $\xi^0 = \xi' = \bm{0}$ and $\xi^m = \xi$.
We note that each $\xi^i$ is feasible because $\xi$ is feasible.
We then deduce from the triangle inequality that
\begin{align}
W_1\left( \mu^{v \gets 0}, \mu^{v \gets 1} \right)
&\le \bE[\| \xi-\xi' \|_1] + \bE\left[ W_1\left( \mu^{v \gets 0, \xi}, \mu^{v \gets 1, \xi'} \right) \right] \nonumber\\
&=\bE[\| \xi-\xi' \|_1] + \bE\left[ W_1\left( \mu^{v \gets 0, \xi^m}, \mu^{v \gets 0, \xi^0} \right) \right] \nonumber\\
&\le 2 + \bE\left[ \sum_{i=1}^m  W_1\left( \mu^{v \gets 0, \xi^{i-1}}, \mu^{v \gets 0, \xi^{i}} \right) \right]. \label{eq:only-sim}
\end{align}
For each $i$, the Wasserstein distance $W_1\big( \mu^{v \gets 0, \xi^{i-1}}, \mu^{v \gets 0, \xi^{i}} \big)$ corresponds to a hardcore model on the subgraph $G \setminus (N^*_v \setminus u_{i}) \setminus \bigcup_{j=1}^{i-1} N_{u_j}$ under two pinnings $u_i \gets 0$ and $u_i \gets 1$; the pinning on $v$ does not matter since the configuration on $N_v$ is fixed.
In particular, $u_i$ is a simplicial vertex in this subgraph since it is already simplicial in $G\setminus (N^*_v \setminus u_i)$ by \cref{lem:claw-free-fact}.
Thus, \cref{eq:only-sim} shows that it suffices to consider the case where the disagreeing vertex is simplicial.

We show in the next proposition that the Wasserstein distance is constant between Gibbs distributions under different pinnings on a simplicial vertex.

\begin{proposition}\label{prop:W1-sim-cf}
Work under the setting of \cref{thm:hardcore-claw-free}.
If $v$ is a simplicial vertex, then we have
\begin{align*}
W_1\left( \mu^{v\gets 0}, \mu^{v\gets 1} \right) \le \Delta\lambda_{\max}.
\end{align*}
\end{proposition}

We can then deduce \cref{prop:hardcore-claw-free-spec-indep} from \cref{prop:W1-sim-cf} by the arguments above.
\begin{proof}[Proof of \cref{prop:hardcore-claw-free-spec-indep}]
Combining \cref{eq:only-sim,prop:W1-sim-cf}, we deduce that
\[
W_1(\mu^{v \gets 0}, \mu^{v \gets 1})
\le
2 +
2 \Delta \lambda_{\max}.
\]
Spectral independence then follows from \cref{lem:coup-indep-imply-spec-indep} and the bound on $W_1$.
\end{proof}

\begin{algorithm}[t]
\caption{Coupling procedure for the hardcore model on claw-free graphs}
\label{alg:couple-hardcore-claw-free}
\begin{algorithmic}[1]
\Procedure{Couple}{$G,\blam,v$}
\State \textbf{Input:} $G=(V,E)$ a claw-free graph, $\blam \in \bR_{>0}^V$ a vector of fugacity, $v\in V$ a simplicial vertex which is disagreeing
\State \textbf{Output:} A pair of random configurations $(\sigma,\sigma') \in 2^{V \setminus v} \times 2^{V \setminus v}$ drawn from a coupling between $\mu_{G,\blam}^{v \gets 0}$ and $\mu_{G,\blam}^{v \gets 1}$

\medskip
\If{$v$ is isolated}
\State Sample $\sigma \sim \mu_{G,\blam}^{v \gets 0}$
\State \Return $(\sigma,\sigma)$
\label{line:iso}
\EndIf

\State Sample $\sigma_{N_v}\sim \big( \mu_{G,\blam}^{v \gets 0} \big)_{N_v}$
and $\sigma'_{N_v} \gets \bm{0}$

\If{$\sigma_{N_v} = \bm{0}$}
\label{line:alg-couple-cf:coupled-no-edge}
\State Sample $\sigma_{V \setminus N^*_v} \sim \mu_{G,\blam}^{v \gets 0, N_v \gets \bm{0}}$
\State \Return $(\sigma,\sigma)$
\label{line:alg-couple-cf:coupled}

\Else\Comment{We must have $\|\sigma_{N_v}\|_1 = 1$}
\label{line:alg-couple-cf:uncoupled-edge}
\State Let $u \in N_v$ such that $\sigma_u = 1$
\State $\big( \sigma'_{V \setminus (N^*_v \setminus u)}, \sigma_{V \setminus (N^*_v \setminus u)} \big) \gets \textsc{Couple}\big( G \setminus (N^*_v \setminus u), \blam_{V \setminus (N^*_v \setminus u)}, u \big)$
\Comment{\cref{lem:claw-free-fact}}
\label{line:alg-call-fact}
\State \Return $(\sigma,\sigma')$
\EndIf

\EndProcedure
\end{algorithmic}
\end{algorithm}

It remains to prove \cref{prop:W1-sim-cf}, which is again proved by a recursive coupling.
\begin{proof}[Proof of \cref{prop:W1-sim-cf}]

We construct a coupling between $\mu_{G,\blam}^{v\gets 0}$ and $\mu_{G,\blam}^{v\gets 1}$ for configurations on $V \setminus v$ using \cref{alg:couple-hardcore-claw-free}, and upper bound the $W_1$ distance via claw-freeness of the graph.

\paragraph{Coupling.} We prove that \cref{alg:couple-hardcore-claw-free} produces a valid coupling by induction on the number of vertices in $G$.
If $v$ is isolated then the configuration on $V \setminus v$ is independent of $\sigma_v$, and hence $\mu_{G,\blam}^{v\gets 0} = \mu_{G,\blam}^{v\gets 1}$; this justifies \cref{line:iso} and also the base case for our induction.
For non-isolated $v$, we sample $\sigma_{N_v} \sim \big( \mu^{v\gets 0}_{G,\blam} \big)_{N_v}$ and $\sigma'_{N_v} \sim \big( \mu^{v\gets 1}_{G,\blam} \big)_{N_v}$; note that we must have $\sigma'_{N_v} = \bm{0}$ for the latter since $v$ is occupied.
If $\sigma_{N_v} = \bm{0} = \sigma'_{N_v}$, then the configuration on the remaining graph is independent of $\sigma_v$, namely $\mu_{G,\blam}^{v \gets 0, N_v \gets \bm{0}} = \mu_{G,\blam}^{v \gets 1, N_v \gets \bm{0}}$,
which justifies \cref{line:alg-couple-cf:coupled}.
Otherwise, we have $\sigma_{N_v} \neq \bm{0}$.
Since $v$ is simplicial, $N_v$ is a clique and hence there is exactly one vertex in $N_v$ that is occupied under $\sigma_{N_v}$, which we denote by $u$.
Let $A = N^*_v \setminus u$ be the closed neighborhood at $v$ excluding $u$, and we have $\sigma_A = \bm{0}$.
By the induction hypothesis we have in \cref{line:alg-call-fact} that
\begin{align*}
\sigma_{V\setminus A} &\sim \mu_{G \setminus A,\blam_{V \setminus A}}^{u \gets 1} = \mu_{G,\blam}^{v \gets 0, N_v \gets \sigma_{N_v}}; \\
\sigma'_{V\setminus A} &\sim \mu_{G \setminus A,\blam_{V \setminus A}}^{u \gets 0} = \mu_{G,\blam}^{v \gets 1, N_v \gets \sigma'_{N_v}}.
\end{align*}
Notice that we can recursively call \cref{alg:couple-hardcore-claw-free} on the input $( G \setminus A, \blam_{V \setminus A}, u )$ because $G \setminus A$ is a claw-free graph and $u$ is a simplicial vertex by \cref{lem:claw-free-fact}.
Thus, $(\sigma_{V\setminus A}, \sigma'_{V\setminus A})$ comes from the desired conditional distributions and therefore the output of \cref{alg:couple-hardcore-claw-free} is from a coupling of $\mu_{G,\blam}^{v\gets 0}$ and $\mu_{G,\blam}^{v\gets 1}$ by induction.

\paragraph{$W_1$ distance.}
Next, we bound the expected $\ell_1$ distance under the coupling generated by \cref{alg:couple-hardcore-claw-free} by induction on the number of vertices.
Base case is when $v$ is isolated, in which case the $\ell_1$ distance is $0$.
Now suppose that $v$ is not isolated.
Observe that in one run of \cref{alg:couple-hardcore-claw-free}: either $\sigma_{N_v} = \bm{0}$ and the $\ell_1$ distance is $0$, or $\|\sigma_{N_v}\| = 1$ and it recursively calls \cref{alg:couple-hardcore-claw-free} in \cref{line:alg-call-fact} on a smaller instance so that the combined expected $\ell_1$ distance is at most $1+\Delta \lambda_{\max}$ by our induction hypothesis ($1$ for the discrepancy at $u$ and $\Delta\lambda_{\max}$ for the recursive call).
We note that the first case happens with probability exactly $\mu_{G,\blam}^{v\gets 0}(\sigma_{N_v} = \bm{0})$, the probability that all neighbors of $v$ are unoccupied.
Therefore, the expected $\ell_1$ distance produced by the root call is at most
\begin{align}\label{eq:cf-W1-rec}
\left( 1 - \mu_{G,\blam}^{v\gets 0}(\sigma_{N_v} = \bm{0}) \right) (1+\Delta\lambda_{\max}).
\end{align}
It remains to lower bound $\mu_{G,\blam}^{v\gets 0}(\sigma_{N_v} = \bm{0})$.
Since $v$ is simplicial, every feasible configuration $\xi$ on $N_v$ satisfies $\xi = \bm{0}$ or $\|\xi\|_1=1$.
We have
\begin{align*}
1 = \mu_{G,\blam}^{v\gets 0}(\sigma_{N_v} = \bm{0}) + \sum_{\xi \in 2^{N_v}:\, \|\xi\|_1 = 1} \mu_{G,\blam}^{v\gets 0}(\sigma_{N_v} = \xi)
\le (1+\Delta\lambda_{\max}) \cdot \mu_{G,\blam}^{v\gets 0}(\sigma_{N_v} = \bm{0}),
\end{align*}
and thus $\mu_{G,\blam}^{v\gets 0}(\sigma_{N_v} = \bm{0}) \ge (1+\Delta\lambda_{\max})^{-1}$.
Plugging into \cref{eq:cf-W1-rec} finishes the proof.
\end{proof}

\subsection{Marginal bounds}
In this subsection we give marginal bounds that are needed.

\begin{proposition}[Marginal boundedness] \label{prop:hardcore-claw-free-marginal-bound}
  Work under the setting of \cref{thm:hardcore-claw-free}.
  For any pinning $\tau$ on a subset $\Lambda\subseteq V$ and any vertex $v\in V\backslash \Lambda$ that can be occupied, we have
  \begin{align*}
    \mu^\tau(\sigma_v=0) \ge \frac 1{1+\lambda_{\max}}
    \quad\text{and}\quad
    \mu^\tau(\sigma_v=1) \ge \frac{\lambda_{\min}}{(1+\lambda_{\min}) \left(1+\Delta \lambda_{\max} + \binom \Delta 2\lambda_{\max}^2\right)}.
  \end{align*}
\end{proposition}
\begin{proof}
  For any configuration $\sigma_1\in 2^{V\backslash (\Lambda\cup v)}$ with non-zero probability under $\mu^\tau$, we have
  \begin{align*}
    \mu^\tau(\sigma_1, \sigma_v=1) \le \lambda_v \mu^\tau(\sigma_1, \sigma_v=0).
  \end{align*}
  Therefore
  \begin{align*}
    \mu^\tau(\sigma_v=0) = \sum_{\sigma_1 \in 2^{V\backslash (\Lambda\cup v)}} \mu^\tau(\sigma_1, \sigma_v=0)
    \ge \frac 1{1+\lambda_v} \sum_{\sigma_1 \in 2^{V\backslash (\Lambda\cup v)}} \mu^\tau(\sigma_1) = \frac 1{1+\lambda_v} \ge \frac 1{1+\lambda_{\max}}.
  \end{align*}
  This proves the first part.

  For the second part, note that for claw-free graphs, at most two neighboring vertices of $v$ can be occupied at the same time, and hence we have
  \begin{align*}
    1 = \sum_{\xi \in 2^{N_v}:\, \|\xi\|_1 \le 2}\mu^\tau(\sigma_{N_v} = \xi)
    \le \left(1+\Delta \lambda_{\max} + \binom \Delta 2\lambda_{\max}^2\right) \mu^\tau(\sigma_{N_v}=\bm{0}).
  \end{align*}
  So $\mu^\tau(\sigma_{N_v}=\bm{0}) \ge \left(1+\Delta \lambda_{\max} + \binom \Delta 2\lambda_{\max}^2\right)^{-1}$ and
  it follows that
  \begin{align*}
    \mu^\tau(\sigma_v=1) &\ge \mu^\tau(\sigma_v=1, \sigma_{N_v}=\bm{0}) \\
    &= \frac{\lambda_v}{1+\lambda_v} \cdot \mu^\tau(\sigma_{N_v}=\bm{0}) \\
    &\ge \frac{\lambda_{\min}}{(1+\lambda_{\min}) \left(1+\Delta \lambda_{\max} + \binom \Delta 2\lambda_{\max}^2\right)}.
  \end{align*}
  This proves the second part.
\end{proof}

We are now ready to prove \cref{thm:hardcore-claw-free}.
\begin{proof}[Proof of \cref{thm:hardcore-claw-free}]
  Follows from \cref{prop:hardcore-claw-free-spec-indep,prop:hardcore-claw-free-marginal-bound,thm:spec-indep-imply-mixing}.
\end{proof}

\bibliographystyle{alpha}
\bibliography{ref}

\newcommand{\etalchar}[1]{$^{#1}$}
\begin{thebibliography}{GGGHP22}

\bibitem[AJK{\etalchar{+}}22]{AJKPV22}
Nima Anari, Vishesh Jain, Frederic Koehler, Huy~Tuan Pham, and Thuy-Duong
  Vuong.
\newblock Entropic independence: optimal mixing of down-up random walks.
\newblock In {\em Proceedings of the 54th Annual ACM SIGACT Symposium on Theory
  of Computing (STOC)}, pages 1418--1430, 2022.

\bibitem[ALO20]{ALO20}
Nima Anari, Kuikui Liu, and Shayan {Oveis Gharan}.
\newblock Spectral independence in high-dimensional expanders and applications
  to the hardcore model.
\newblock In {\em Proceedings of the 61st Annual IEEE Symposium on Foundations
  of Computer Science (FOCS)}, pages 1319--1330, 2020.

\bibitem[Bar16]{Bar16book}
Alexander Barvinok.
\newblock {\em Combinatorics and Complexity of Partition Functions}, volume~30.
\newblock Springer Algorithms and Combinatorics, 2016.

\bibitem[BCC{\etalchar{+}}22]{BCCPSV22}
Antonio Blanca, Pietro Caputo, Zongchen Chen, Daniel Parisi, Daniel
  {\v{S}}tefankovi{\v{c}}, and Eric Vigoda.
\newblock On mixing of {M}arkov chains: coupling, spectral independence, and
  entropy factorization.
\newblock {\em Electronic Journal of Probability}, 27:1--42, 2022.

\bibitem[BCR21]{BCR21}
Ferenc Bencs, P{\'e}ter Csikv{\'a}ri, and Guus Regts.
\newblock Some applications of {W}agner's weighted subgraph counting
  polynomial.
\newblock {\em The Electronic Journal of Combinatorics}, 28(4), 2021.

\bibitem[BR09]{BR09}
Ivona Bez{\'{a}}kov{\'{a}} and William~A. Rummler.
\newblock Sampling edge covers in 3-regular graphs.
\newblock In {\em Proceedings of the 34th International Symposium on
  Mathematical Foundations of Computer Science (MFCS)}, pages 137--148, 2009.

\bibitem[Br{\"a}15]{Bra15}
Petter Br{\"a}nd{\'e}n.
\newblock Unimodality, log-concavity, real-rootedness and beyond.
\newblock {\em Handbook of enumerative combinatorics}, 87:437, 2015.

\bibitem[CE22]{CE22}
Yuansi Chen and Ronen Eldan.
\newblock Localization schemes: {A} framework for proving mixing bounds for
  {M}arkov chains.
\newblock In {\em Proceedings of the 63rd Annual IEEE Symposium on Foundations
  of Computer Science (FOCS)}, pages 110--122. IEEE, 2022.

\bibitem[CFYZ21]{CFYZ21}
Xiaoyu Chen, Weiming Feng, Yitong Yin, and Xinyuan Zhang.
\newblock Rapid mixing of {G}lauber dynamics via spectral independence for all
  degrees.
\newblock In {\em Proceedings of the 62nd Annual IEEE Symposium on Foundations
  of Computer Science (FOCS)}, pages 137--148. IEEE, 2021.

\bibitem[CFYZ22]{CFYZ22}
Xiaoyu Chen, Weiming Feng, Yitong Yin, and Xinyuan Zhang.
\newblock Optimal mixing for two-state anti-ferromagnetic spin systems.
\newblock In {\em Proceedings of the 63rd Annual IEEE Symposium on Foundations
  of Computer Science (FOCS)}, pages 588--599. IEEE, 2022.

\bibitem[CLMM23]{CLMM23}
Zongchen Chen, Kuikui Liu, Nitya Mani, and Ankur Moitra.
\newblock Strong spatial mixing for colorings on trees and its algorithmic
  applications.
\newblock {\em arXiv preprint arXiv:2304.01954}, 2023.

\bibitem[CLV21]{chen2021optimal}
Zongchen Chen, Kuikui Liu, and Eric Vigoda.
\newblock Optimal mixing of {G}lauber dynamics: Entropy factorization via
  high-dimensional expansion.
\newblock In {\em Proceedings of the 53rd Annual ACM SIGACT Symposium on Theory
  of Computing}, pages 1537--1550, 2021.

\bibitem[CLV22]{chen2022spectral}
Zongchen Chen, Kuikui Liu, and Eric Vigoda.
\newblock Spectral independence via stability and applications to {H}olant-type
  problems.
\newblock In {\em 2021 IEEE 62nd Annual Symposium on Foundations of Computer
  Science (FOCS)}, pages 149--160. IEEE, 2022.

\bibitem[CMM23]{CMM23}
Zongchen Chen, Nitya Mani, and Ankur Moitra.
\newblock From algorithms to connectivity and back: finding a giant component
  in random $k$-{SAT}.
\newblock In {\em Proceedings of the 2023 Annual ACM-SIAM Symposium on Discrete
  Algorithms (SODA)}, pages 3437--3470. SIAM, 2023.

\bibitem[CS07]{CS07}
Maria Chudnovsky and Paul Seymour.
\newblock The roots of the independence polynomial of a clawfree graph.
\newblock {\em Journal of Combinatorial Theory, Series B}, 97(3):350--357,
  2007.

\bibitem[CZ23]{chen2023near}
Xiaoyu Chen and Xinyuan Zhang.
\newblock A near-linear time sampler for the {I}sing model with external field.
\newblock In {\em Proceedings of the 2023 Annual ACM-SIAM Symposium on Discrete
  Algorithms (SODA)}, pages 4478--4503. SIAM, 2023.

\bibitem[DGM21]{DGM21}
Martin Dyer, Catherine Greenhill, and Haiko M{\"u}ller.
\newblock Counting independent sets in graphs with bounded bipartite pathwidth.
\newblock {\em Random Structures \& Algorithms}, 59(2):204--237, 2021.

\bibitem[DHJM21]{DHJM21}
Martin Dyer, Marc Heinrich, Mark Jerrum, and Haiko M\"{u}ller.
\newblock Polynomial-time approximation algorithms for the antiferromagnetic
  {I}sing model on line graphs.
\newblock {\em Combinatorics, Probability and Computing}, pages 1--17, 2021.

\bibitem[FGW22]{FGW22}
Weiming Feng, Heng Guo, and Jiaheng Wang.
\newblock {S}wendsen-{W}ang dynamics for the ferromagnetic {I}sing model with
  external fields.
\newblock {\em arXiv preprint arXiv:2205.01985}, 2022.

\bibitem[GGGHP22]{GGGH22}
Andreas Galanis, Leslie~Ann Goldberg, Heng Guo, and Andr{\'e}s Herrera-Poyatos.
\newblock Fast sampling of satisfying assignments from random $k$-{SAT}.
\newblock {\em arXiv preprint arXiv:2206.15308}, 2022.

\bibitem[GJ18]{GJ18}
Heng Guo and Mark Jerrum.
\newblock Random cluster dynamics for the {I}sing model is rapidly mixing.
\newblock {\em Annals of Applied Probability}, 28(2):1292--1313, 2018.

\bibitem[GLLZ21]{GLLZ21}
Heng Guo, Chao Liao, Pinyan Lu, and Chihao Zhang.
\newblock Zeros of {H}olant problems: Locations and algorithms.
\newblock {\em ACM Transactions on Algorithms}, 17(1):1--25, 2021.

\bibitem[HLZ16]{HLZ16}
Lingxiao Huang, Pinyan Lu, and Chihao Zhang.
\newblock Canonical paths for {MCMC}: {F}rom art to science.
\newblock In {\em Proceedings of the twenty-seventh annual ACM-SIAM symposium
  on discrete algorithms (SODA)}, pages 514--527. SIAM, 2016.

\bibitem[Jer03]{Jbook}
Mark Jerrum.
\newblock {\em Counting, Sampling and Integrating: Algorithms and Complexity}.
\newblock Lectures in Mathematics, ETH Zürich. Birkhäuser Basel, 2003.

\bibitem[JS89]{JS89}
Mark Jerrum and Alistair Sinclair.
\newblock Approximating the permanent.
\newblock {\em SIAM journal on computing}, 18(6):1149--1178, 1989.

\bibitem[JS93]{JS93}
Mark Jerrum and Alistair Sinclair.
\newblock Polynomial-time approximation algorithms for the {I}sing model.
\newblock {\em SIAM Journal on Computing}, 22(5):1087--1116, 1993.

\bibitem[JSV04]{JSV04}
Mark Jerrum, Alistair Sinclair, and Eric Vigoda.
\newblock A polynomial-time approximation algorithm for the permanent of a
  matrix with nonnegative entries.
\newblock {\em Journal of the ACM (JACM)}, 51(4):671--697, 2004.

\bibitem[LLL14]{LLL14}
Chengyu Lin, Jingcheng Liu, and Pinyan Lu.
\newblock A simple {FPTAS} for counting edge covers.
\newblock In {\em Proceedings of the 25th Annual ACM-SIAM Symposium on Discrete
  Algorithms (SODA)}, pages 341--348, 2014.

\bibitem[LLZ14]{LLZ14}
Jingcheng Liu, Pinyan Lu, and Chihao Zhang.
\newblock {FPTAS} for counting weighted edge covers.
\newblock In {\em Proceedings of the 22nd Annual European Symposium on
  Algorithms (ESA)}, pages 654--665, 2014.

\bibitem[LR19]{LR19}
Jonathan~D Leake and Nick~R Ryder.
\newblock Generalizations of the matching polynomial to the multivariate
  independence polynomial.
\newblock {\em Algebraic Combinatorics}, 2(5):781--802, 2019.

\bibitem[LSS19]{LSS19}
Jingcheng Liu, Alistair Sinclair, and Piyush Srivastava.
\newblock The {I}sing partition function: Zeros and deterministic
  approximation.
\newblock {\em Journal of Statistical Physics}, 174(2):287--315, 2019.

\bibitem[LWZ14]{LWZ14}
Pinyan Lu, Menghui Wang, and Chihao Zhang.
\newblock {FPTAS} for weighted {F}ibonacci gates and its applications.
\newblock In {\em Proceedings of the 41st International Colloquium on Automata,
  Languages and Programming (ICALP)}, pages 787--799. Springer, 2014.

\bibitem[Mat08]{Mat08}
James Matthews.
\newblock {\em Markov Chains for Sampling Matchings}.
\newblock PhD thesis, University of Edinburgh, 2008.

\bibitem[McQ13]{McQ13}
Colin McQuillan.
\newblock Approximating {H}olant problems by winding.
\newblock {\em ArXiv preprint, arXiv:1301.2880}, 2013.

\bibitem[PR17]{PR17}
Viresh Patel and Guus Regts.
\newblock Deterministic polynomial-time approximation algorithms for partition
  functions and graph polynomials.
\newblock {\em SIAM Journal on Computing}, 46(6):1893--1919, 2017.

\end{thebibliography}
\end{document}